\documentclass[twocolumn,superscriptaddress,amsfont,amssymb,amsmath, showpacs,balancelastpage]{revtex4-1}
\usepackage{CJKutf8}

\usepackage[utf8]{inputenc}
\usepackage{amsmath,amssymb,amsfonts,stmaryrd}
\usepackage{physics}
\usepackage{dsfont}
\usepackage{graphicx}
\usepackage{xcolor}
\usepackage{graphicx}
\usepackage{cancel}
\usepackage{natbib}
\usepackage{color}
\usepackage{tikz}
\usepackage{mathrsfs}
\usepackage{relsize}
\usepackage{multirow}
\usepackage[linktocpage]{hyperref}
\usepackage[all]{xy}
\hypersetup{colorlinks=true,citecolor=blue,linkcolor=blue, urlcolor=blue, breaklinks=true}

\usepackage{float}

\usepackage{bm} 
\usepackage{subfigure} 
\usepackage{environ}
\usepackage{url}
\usepackage[margin=0.75in]{geometry}
\usepackage{ulem}

\usepackage{mathtools}

\newcommand{\Z}{{\mathbb Z}}

\NewEnviron{eqs}{%
\begin{equation}\begin{split}
    \BODY
\end{split}\end{equation}}

\definecolor{dkgreen}{rgb}{0,0.5,0}

\usepackage{amsthm}
\newtheorem{theorem}{Theorem}

\begin{document}

\begin{CJK*}{UTF8}{bsmi}

\title{Exact $\mathbb{Z}_2$ electromagnetic duality of $\mathbb{Z}_2$ toric code is non-Clifford}

\author{Ryohei Kobayashi}
\email[E-mail: ]{kobayashir061@gmail.com}
\affiliation{School of Natural Sciences, Institute for Advanced Study, Princeton, NJ 08540, USA}
\affiliation{Department of Applied Physics, The University of Tokyo, Tokyo 113-8656, Japan}

\date{\today}
\begin{abstract}
The 2D $\mathbb{Z}_2$ toric code admits a global symmetry exchanging electric and magnetic quasiparticles, known as electromagnetic duality. Known realizations include lattice translation symmetry, an exact $\mathbb{Z}_4$ symmetry generated by a Clifford circuit, and an exact $\mathbb{Z}_2$ symmetry generated by a non-Clifford circuit. We show that a Clifford electromagnetic duality cannot realize an exact internal $\mathbb{Z}_2$ symmetry. This is proved rigorously for symmetries with coarse translation invariance by $l$ lattice units for generic odd $l$. Therefore an exact internal $\mathbb{Z}_2$ electromagnetic duality must be non-Clifford, whereas generic internal Clifford realization necessarily has $\mathbb{Z}_{2^m}$ algebra with $m\ge 2$. Our result suggests an unexpected connection between the algebra of exact electromagnetic duality and Clifford hierarchy of circuits.
\end{abstract}

\maketitle

\end{CJK*}




\textit{Introduction.-- }
Global symmetry plays a central role in quantum many-body systems. It constrains dynamics, and organizes phases of matter. In topologically ordered systems, symmetry acts nontrivially on anyons and generates logical operators of error-correcting codes. Among solvable models of topological order, the $\mathbb{Z}_2$ toric code occupies a distinguished place \cite{Kitaev2003toriccode}: it is the canonical lattice realization of $\mathbb{Z}_2$ gauge theory, and its global symmetries provide a basic platform for both quantum error correction and fault-tolerant logical operations \cite{Yoshida_gate_SPT_2015, Yoshida:2015cia, Yoshida2017387, Barkeshli:2022wuz, Barkeshli:2022edm, Barkeshli:2023bta}.

At the level of topological order, (2+1)D $\mathbb{Z}_2$ gauge theory admits an electromagnetic duality exchanging the electric and magnetic quasiparticles. In the square-lattice toric code this duality exchanges star and plaquette operators. A standard microscopic realization combines a transversal Hadamard transformation with a lattice translation, and therefore is not an exact $\Z_2$ symmetry. More recently, an exact internal Clifford symmetry exchanging $e$ and $m$ was constructed as a finite-depth Clifford circuit generating a $\mathbb{Z}_4$ action rather than a $\mathbb{Z}_2$ action \cite{Barkeshli:2022wuz}. On the other hand, exact non-Clifford $\Z_2$ symmetry for electromagnetic duality are known to exist on the microscopic Hilbert space \cite{shirley2025QCA, Tu_2026}. Thus the obstruction to exact $\Z_2$ electromagnetic duality is not absolute; rather, it is specific to Clifford operators.

In this Letter, we prove a no-go theorem for exact \textit{Clifford} $\Z_2$ internal symmetry for electromagnetic duality. More precisely, we show that there is no exact Clifford $\mathbb{Z}_2$ symmetry of the square-lattice toric code that exchanges $e$ and $m$ and is translation-invariant with respect to an enlarged $l\times l$ unit cell, with generic odd integer $l$. Equivalently, if an exact symmetry implementing electromagnetic duality is represented by an automorphism of the Pauli group that commutes with translations by $l\hat x$ and $l\hat y$ for odd $l$, then it cannot have order two.

While our proof is established for odd $l$, we also checked the first even cases, $l=2,4$, by computer and found no exact $\Z_2$ symmetry exchanging $e$ and $m$ with local action on Pauli operators. Although this does not amount to a proof for all even $l$, it provides further evidence that exact $\Z_2$ electromagnetic duality is absent within Clifford operators, and hints at an interesting connection between the algebra of exact electromagnetic duality and Clifford hierarchy of circuits.

Our proof is formulated in the polynomial formalism for translation-invariant Pauli stabilizer Hamiltonians~\cite{Haah2013polynomial}. In this language, a Clifford symmetry with $l\times l$ enlarged translation symmetry is represented by a symplectic automorphism over the blocked Laurent polynomial ring. We show that, for odd $l$, the conditions of symplecticity, order two, and electromagnetic duality are algebraically incompatible. As a consequence, any exact odd-supercell translation-invariant Clifford implementation of electromagnetic duality must have order at least four. This lower bound is saturated by the explicit $\mathbb{Z}_4$ symmetry obtained in Ref.~\cite{Barkeshli:2022wuz, shirley2025QCA} (see Fig.~\ref{fig:Z4}).

\begin{figure}[htb]
\centering
\includegraphics[width=0.35\textwidth]{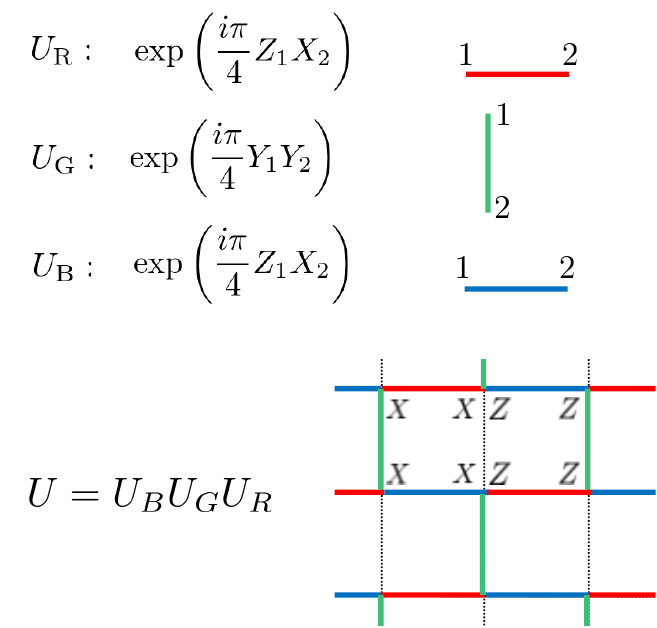}
\caption{A finite-depth circuit $U=U_B U_G U_R$ of the $\mathbb{Z}_2$ toric code generates an exact $\mathbb{Z}_4$ Clifford symmetry exchanging $e$ and $m$ \cite{Barkeshli:2022wuz}. On the checkerboard lattice, each plaquette supports a toric code stabilizer, as shown in the lower-right. Plaquettes of one color support $X$ stabilizers, while plaquettes of the other color support $Z$ stabilizers. Under conjugation by $U$, the $X$ and $Z$ stabilizers are exchanged. The $\mathbb{Z}_4$ algebra of this operator is discussed explicitly in Ref.~\cite{shirley2025QCA}.}
\label{fig:Z4}
\end{figure}

\textit{Toric code and polynomial formalism.--}
We begin by briefly reviewing the $\mathbb{Z}_2$ toric code on a square lattice and its description in the polynomial formalism for translation-invariant Pauli stabilizer Hamiltonians. We follow the notation of Ref.~\cite{Haah2013polynomial}. See e.g., \cite{Haah2021classification, Haah2011haahcode, Haah2022QCA, Tantivasadakarn2020JW, liang2025operatoralgebraalgorithmicconstruction, Chen2024fermionic, sun2026cliffordquantumcellularautomata, ruba2025wittgroupsbulkboundarycorrespondence} for the applications of polynomial formalisms in Pauli models.

The toric code is a commuting-projector Hamiltonian of qubits placed on the edges of a square lattice,
\begin{equation}
H_{\mathrm{TC}}
=
-\sum_{v} A_v-\sum_{p} B_p,
\label{eq:TC-H}
\end{equation}
where
\begin{equation}
A_v=\prod_{v\subset\partial e} X_e~,
\qquad
B_p=\prod_{e\subset \partial p} Z_e~.
\label{eq:TC-stabilizers}
\end{equation}

We then express the toric code utilizing the presentation of the translation-invariant Pauli Hamiltonian by Laurent polynomials. We consider one qubit at each edge of the square lattice. A lattice translation by one unit in the $\hat x$ or $\hat y$ direction is represented by multiplication by formal variables $x$ or $y$, respectively.

The translation group is thus encoded by the Laurent polynomial ring
\begin{equation}
R=\mathbb{F}_2[x^{\pm1},y^{\pm1}] .
\end{equation}

We represent generic Pauli operator with finite support by a column vector with entries in $R$. We first write single Pauli $X,Z$ operator by the following vectors:
\begin{align}
    X_{12} = \begin{pmatrix}
1\\
0 \\
0 \\
0
\end{pmatrix}, \
    X_{14} = \begin{pmatrix}
0\\
1 \\
0 \\
0
\end{pmatrix}, \
 Z_{12} = \begin{pmatrix}
0\\
0 \\
1 \\
0
\end{pmatrix}, \ 
    Z_{14} = \begin{pmatrix}
0\\
0 \\
0 \\
1
\end{pmatrix},
\end{align}
where $12,14$ are edges of the square lattice presented in Fig.~\ref{fig:lattice}.
We then express generic Pauli operators with finite support by a vector
\begin{equation}
v=
\begin{pmatrix}
v_X\\
v_Z
\end{pmatrix}
\in R^4,
\qquad
v_X,v_Z\in R^2,
\end{equation}
where the upper two entries encode the product of $X$ operators involved in a single operator, and the lower two entries encode $Z$ operators. The sums of monomials in entries of $v_X, v_Z$ indicate products of Pauli operators over translated cells, and multiplying the monomial $x^m y^n$ to $v$ translates the operator by a vector $(m,n)$.
For instance, the toric code Hamiltonians nearby the edges $12,14$ correspond to
\begin{align}
    A_1 = \begin{pmatrix}
1+x^{-1} \\
1+y^{-1} \\
0 \\
0 
\end{pmatrix}, \ 
    B_{1245} = \begin{pmatrix}
0 \\
0 \\
1+y \\
1+x 
\end{pmatrix}.
\end{align}

A basic ingredient in this formalism is the antipode map denoted by
\begin{equation}
\overline{x}=x^{-1},\qquad \overline{y}=y^{-1},
\end{equation}
extended linearly to all of $R$. Thus, for
\begin{equation}
f=\sum_{m,n} c_{m,n} x^m y^n \in R,
\label{eq:def of f}
\end{equation}
its antipode is
\begin{equation}
\bar f=\sum_{m,n} c_{m,n} x^{-m} y^{-n}.
\end{equation}
If $v$ is a column vector over $R$, we define
\begin{equation}
v^\dagger := \bar v^{\,T},
\end{equation}
namely transpose followed by the antipode map.

The commutation of Pauli operators is encoded by the symplectic form
\begin{equation}
\Lambda=
\begin{pmatrix}
0&I_2\\
I_2&0
\end{pmatrix}.
\end{equation}
Two Pauli operators represented by $u,v\in R^4$ commute iff
\begin{equation}
\langle u^\dagger \Lambda v\rangle_0 =0 ,
\label{eq:commute-rule}
\end{equation}
where $\langle f\rangle_0$ with $f$ given by \eqref{eq:def of f} is defined as $\langle f\rangle_0:= c_{0,0}$.
A translation-invariant stabilizer Hamiltonian with $t$ generator types is then specified by a matrix
\begin{equation}
\sigma \in \mathrm{Mat}_{4\times t}(R),
\end{equation}
whose columns represent the stabilizer generators in one reference unit cell. The $\Z_2$ toric code is represented by
\begin{equation}
\sigma_{\mathrm{TC}}
=
\begin{pmatrix}
1+\bar x & 0\\
1+\bar y & 0\\
0 & 1+y\\
0 & 1+x
\end{pmatrix}.
\label{eq:sigma-tc-short}
\end{equation}
The commuting-projector condition is simply
\begin{equation}
\sigma^\dagger \Lambda \sigma =0 .
\label{eq:stabilizer-commuting}
\end{equation}

This formulation is especially convenient for analyzing translation-invariant Clifford symmetries. Any such symmetry is represented by a symplectic automorphism
\begin{equation}
Q \in \mathrm{Sp}_{4}(R),
\qquad
Q^\dagger \Lambda Q = \Lambda,
\label{eq:symplectic-Q}
\end{equation}
acting on the Pauli module $P=R^4$. 
The exact $\Z_2$ electromagnetic duality of the toric code is generated by a unitary $U$ satisfying $U^2=1$, that exchanges $X$ and $Z$ stabilizer by
\begin{align}
    U A_v U^\dagger = B_{u(v)}~, UB_pU^\dagger = A_{\tilde u(p)}~,
\end{align}
where $u$ is the lattice translation by a vector $(m+1/2, n+1/2)$ with integers $m,n$ that maps a vertex $v$ to a plaquette $p=u(v)$. $\tilde u$ is the inverse translation by $-(m+1/2, n+1/2)$ to be compatible with the $\Z_2$ algebra $U^2=1$.

Such electromagnetic duality corresponds to $Q$ acting by interchanging the two columns of $\sigma_{\mathrm{TC}}$ up to a translation factor. In the following, we use the polynomial formalism with the enlarged $l\times l$ unit cell to prove the absence of an exact Clifford $\mathbb{Z}_2$ electromagnetic duality with the translation symmetry by $l$ unit lattice translations.

\begin{figure}[htb]
\centering
\includegraphics[width=0.15\textwidth]{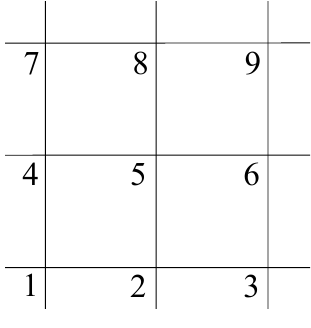}
\caption{A square lattice with a single qubit on each edge.
}
\label{fig:lattice}
\end{figure}

\textit{Blocked polynomial formalism.--}
To analyze Clifford symmetries that need only commute with translations by $l$ lattice spacings, we pass to an $l\times l$ enlarged unit cell. Let
\begin{equation}
R_l=\mathbb{F}_2[X^{\pm1},Y^{\pm1}],\qquad X=x^l,\quad Y=y^l,
\end{equation}
and regard
\begin{equation}
R_1=\mathbb{F}_2[x^{\pm1},y^{\pm1}]
\end{equation}
as a free $R_l$-module of rank $l^2$. Choosing the $R_l$-basis
\begin{equation}
\mathbf b=(x^i y^j)_{0\le i,j<l},
\end{equation}
multiplication by the fine translations $x$ and $y$ is represented by matrices
\begin{equation}
T_x,T_y\in \mathrm{Mat}_{l^2}(R_l)
\end{equation}
defined by
\begin{equation}
\mathbf b\, T_x = x\,\mathbf b,\qquad
\mathbf b\, T_y = y\,\mathbf b.
\end{equation}
These satisfy
\begin{equation}
T_x^\dagger=T_x^{-1},\qquad T_y^\dagger=T_y^{-1}.
\end{equation}

With $2l^2$ qubits in the enlarged unit cell, a Pauli operator is a vector of the vector space
\begin{equation}
P_l=R_l^{4l^2},
\end{equation}
equipped with symplectic form
\begin{equation}
\Lambda_l=
\begin{pmatrix}
0&I_{2l^2}\\
I_{2l^2}&0
\end{pmatrix}.
\end{equation}
The blocked toric-code stabilizer matrix takes the form
\begin{equation}
\sigma_{\mathrm{TC}}^{(l)}
=
\begin{pmatrix}
d & 0\\
0 & \tilde d
\end{pmatrix},
\end{equation}
where
\begin{equation}
d=
\begin{pmatrix}
1+T_x^\dagger\\
1+T_y^\dagger
\end{pmatrix},
\qquad
\tilde d=
\begin{pmatrix}
1+T_y\\
1+T_x
\end{pmatrix}.
\label{eq:dtilde-blocked}
\end{equation}
Here each displayed block is an $l^2\times l^2$ matrix over $R_l$.

\textit{Absence of exact Clifford $\mathbb{Z}_2$ electromagnetic duality for odd supercells.--}
We now state and prove the main result.

\begin{theorem}
Let $l$ be odd. There exists no Clifford $\mathbb{Z}_2$ symmetry of the square-lattice toric code Hamiltonian that
\begin{enumerate}
\item is invariant under translations by $l\hat x$ and $l\hat y$, and
\item implements electromagnetic duality by exchanging star and plaquette stabilizers.
\end{enumerate}
Equivalently, for odd $l$ there is no matrix
\begin{equation}
Q\in \mathrm{Mat}_{4l^2}(R_l)
\end{equation}
such that
\begin{align}
Q^\dagger \Lambda_l Q &= \Lambda_l, \label{eq:Q-symplectic-blocked}\\
Q^2 &= I, \label{eq:Q-involution-blocked}\\
Q \sigma_{\mathrm{TC}}^{(l)} &= \sigma_{\mathrm{TC}}^{(l)} S_U, \label{eq:Q-duality-blocked}
\end{align}
where
\begin{equation}
S_U=
\begin{pmatrix}
0 & U\\
U^{-1} & 0
\end{pmatrix}
\end{equation}
and $U\in \mathrm{Mat}_{l^2}(R_l)$ is the blocked representation of multiplication by a monomial $u=x^a y^b$.
\end{theorem}

\begin{proof}
Write
\begin{equation}
Q=
\begin{pmatrix}
A & B\\
C & D
\end{pmatrix},
\qquad
A,B,C,D\in \mathrm{Mat}_{2l^2}(R_l),
\end{equation}
with $d,\tilde d$ defined in Eq.~\eqref{eq:dtilde-blocked}. Then Eq.~\eqref{eq:Q-duality-blocked} is equivalent to
\begin{align}
A d &= 0, \label{eq:Ad-blocked}\\
C d &= \tilde d\, U^{-1}, \label{eq:Cd-blocked}\\
B \tilde d &= d\, U, \label{eq:Bd-blocked}\\
D \tilde d &= 0. \label{eq:Dd-blocked}
\end{align}

Since $Q$ is symplectic, Eq.~\eqref{eq:Q-symplectic-blocked} implies
\begin{equation}
Q^{-1}=\Lambda_l Q^\dagger \Lambda_l.
\end{equation}
Combining this with the involution condition \eqref{eq:Q-involution-blocked}, we obtain
\begin{equation}
Q=\Lambda_l Q^\dagger \Lambda_l
=
\begin{pmatrix}
D^\dagger & B^\dagger\\
C^\dagger & A^\dagger
\end{pmatrix}.
\end{equation}
Hence
\begin{equation}
D=A^\dagger,\qquad B=B^\dagger,\qquad C=C^\dagger.
\label{eq:block-hermitian-blocked}
\end{equation}
In particular, any solution must contain a Hermitian matrix $C$ satisfying Eq.~\eqref{eq:Cd-blocked}. We now show that no such $C$ exists when $l$ is odd.

A particular solution of Eq.~\eqref{eq:Cd-blocked} is
\begin{equation}
C_0=
\begin{pmatrix}
0 & T_y U^{-1}\\
U^{-1} T_x & 0
\end{pmatrix},
\end{equation}
since all multiplication operators commute and
\begin{equation}
C_0 d=
\begin{pmatrix}
T_y U^{-1}(1+T_y^\dagger)\\
U^{-1}T_x(1+T_x^\dagger)
\end{pmatrix}
=
\begin{pmatrix}
(1+T_y)U^{-1}\\
(1+T_x)U^{-1}
\end{pmatrix}
=
\tilde d\, U^{-1}.
\end{equation}
Generic solution to Eq.~\eqref{eq:Cd-blocked} is obtained by shifting $C_0$ by a solution $M=(M_1,M_2)$ of 
the homogeneous equation $(M_1, M_2)d=0,$ i.e.,
\begin{equation}
M_1(1+T_x^\dagger)+M_2(1+T_y^\dagger)=0,
\end{equation}
which has general solution
\begin{equation}
(M_1,M_2)=L(1+T_y^\dagger,\ 1+T_x^\dagger)=L\tilde d^\dagger,
\label{eq:general solution M1,M2}
\end{equation}
using $L\in \text{Mat}_{2l^2\times l^2}(R_l)$.
This is understood by focusing on each $1\times l^2$ row vector $m_1,m_2$ of $M_1, M_2$ and evaluate the equation on $\mathbf{b}^\dagger$,
\begin{equation}
[m_1(1+T_x^\dagger)+m_2(1+T_y^\dagger)]\mathbf{b}^\dagger=0.
\end{equation}
Since $T_x^\dagger \mathbf{b}^\dagger=x^{-1}\mathbf{b}^\dagger$, $T_y^\dagger \mathbf{b}^\dagger=y^{-1}\mathbf{b}^\dagger$, we get
\begin{align}
(1+x^{-1}) m_1\mathbf{b}^\dagger + (1+y^{-1}) m_2\mathbf{b}^\dagger = 0. 
\end{align}
The general solution to the above equation is  given by $(m_1\mathbf{b}^\dagger, m_2\mathbf{b}^\dagger) = g(1+y^{-1}, 1+x^{-1})$ using a polynomial $g\in R_1$, hence general form of $(m_1,m_2)$ is
\begin{align}
    (m_1,m_2) = l(1+T_y^\dagger, 1+T_x^\dagger)
\end{align}
using a row vector $l$, which leads to Eq.~\eqref{eq:general solution M1,M2}. Therefore, general solution of Eq.~\eqref{eq:Cd-blocked} is
\begin{equation}
C=C_0+L\tilde d^\dagger,
\qquad
L=
\begin{pmatrix}
P\\
Q
\end{pmatrix},
\qquad
P,Q\in \mathrm{Mat}_{l^2}(R_l),
\end{equation}
namely
\begin{equation}
C=
\begin{pmatrix}
P(1+T_y^\dagger) & T_yU^{-1}+P(1+T_x^\dagger)\\[1mm]
U^{-1}T_x+Q(1+T_y^\dagger) & Q(1+T_x^\dagger)
\end{pmatrix}.
\label{eq:C-general-blocked}
\end{equation}

We now impose Hermiticity. From $C=C^\dagger$, the diagonal blocks of Eq.~\eqref{eq:C-general-blocked} give
\begin{align}
P(1+T_y^\dagger) &= (1+T_y)P^\dagger, \label{eq:Pdiag}\\
Q(1+T_x^\dagger) &= (1+T_x)Q^\dagger, \label{eq:Qdiag}
\end{align}
while the off-diagonal blocks give
\begin{equation}
T_yU^{-1}+P(1+T_x^\dagger)
=
T_x^\dagger U + (1+T_y)Q^\dagger.
\label{eq:offdiag-blocked}
\end{equation}

To reduce these matrix equations to scalar Laurent-polynomial equations, define
\begin{equation}
\phi:\mathrm{Mat}_{l^2}(R_l)\to R_1,
\qquad
\phi(M):=\mathbf b\, M\, \mathbf b^\dagger.
\end{equation}
This map satisfies
\begin{align*}
\phi(T_x M)&=x\,\phi(M), &
\phi(T_y M)&=y\,\phi(M), \\
\phi(MT_x^\dagger)&=x^{-1}\phi(M), &
\phi(MT_y^\dagger)&=y^{-1}\phi(M), \\
\phi(M^\dagger)&=\overline{\phi(M)}.
\end{align*}
Indeed, these identities follow immediately from the defining relations
\begin{equation}
\mathbf b\,T_x=x\,\mathbf b,\qquad \mathbf b\,T_y=y\,\mathbf b,
\end{equation}
together with
\begin{equation}
T_x^\dagger\mathbf b^\dagger=x^{-1}\mathbf b^\dagger,\qquad
T_y^\dagger\mathbf b^\dagger=y^{-1}\mathbf b^\dagger.
\end{equation}

Set
\begin{equation}
p:=\phi(P),\qquad q:=\phi(Q).
\end{equation}
Applying $\phi$ to Eqs.~\eqref{eq:Pdiag} and \eqref{eq:Qdiag} yields
\begin{equation}
p(1+y^{-1})=(1+y)\bar p,
\qquad
q(1+x^{-1})=(1+x)\bar q.
\end{equation}
Since $1+y^{-1}=y^{-1}(1+y)$ and $1+x^{-1}=x^{-1}(1+x)$, this is equivalent to
\begin{equation}
p=y\bar p,\qquad q=x\bar q.
\label{eq:pq-fixed-blocked}
\end{equation}

Next, applying $\phi$ to Eq.~\eqref{eq:offdiag-blocked} gives
\begin{equation}
\phi(T_yU^{-1}) + p(1+x^{-1})
=
\phi(T_x^\dagger U) + (1+y)\bar q.
\label{eq:offdiag-phi}
\end{equation}
Let
\begin{equation}
\nu:=\phi(I)=\mathbf b\,\mathbf b^\dagger=l^2 \pmod 2.
\end{equation}
Since $U$ is the blocked matrix representing multiplication by $u=x^a y^b$, one has
\begin{equation}
\phi(U)=u\,\mathbf b\mathbf b^\dagger=\nu\,u,
\qquad
\phi(U^{-1})=\nu\,u^{-1}.
\end{equation}
Hence Eq.~\eqref{eq:offdiag-phi} becomes
\begin{equation}
\nu\, y u^{-1}+p(1+x^{-1})
=
\nu\, x^{-1}u +(1+y)\bar q.
\end{equation}
Multiplying by $x$ and using $q=x\bar q$ from Eq.~\eqref{eq:pq-fixed-blocked}, we obtain
\begin{equation}
(1+x)p+(1+y)q=\nu\,(u+xy\,u^{-1}).
\label{eq:key-eqn-blocked}
\end{equation}

Now assume $l$ is odd. Then $\nu=l^2 \equiv 1 \pmod 2$, and Eq.~\eqref{eq:key-eqn-blocked} reduces to
\begin{equation}
(1+x)p+(1+y)q=u+xy\,u^{-1},
\end{equation}
with $p,q$ constrained by Eq.~\eqref{eq:pq-fixed-blocked}.  Define
\begin{align*}
P_{m,n} &:= x^m y^n + x^{-m} y^{1-n}, \\
Q_{m,n} &:= x^m y^n + x^{1-m} y^{-n}, \\
E_{m,n} &:= x^m y^n + x^{1-m} y^{1-n}.
\end{align*}
The relations $p=y\bar p$ and $q=x\bar q$ imply that $p$ and $q$ are finite $\mathbb{F}_2$-linear combinations of the $P_{m,n}$ and $Q_{m,n}$, respectively. Moreover,
\begin{align}
(1+x)P_{m,n} &= E_{m,n}+E_{m+1,n},\\
(1+y)Q_{m,n} &= E_{m,n}+E_{m,n+1}.
\end{align}
Thus every summand on the left-hand side of Eq.~\eqref{eq:key-eqn-blocked} contributes an even number of $E$-terms. Writing
\begin{equation}
\pi\!\left(\sum_{m,n}\alpha_{m,n}E_{m,n}\right)
:=
\sum_{m,n}\alpha_{m,n}\quad (\mathrm{mod}\ 2),
\end{equation}
we obtain
\begin{equation}
\pi\!\left((1+x)p+(1+y)q\right)=0.
\label{eq:LHS-parity-blocked}
\end{equation}
On the other hand, since $u=x^a y^b$ is a monomial,
\begin{equation}
u+xy\,u^{-1}
=
x^a y^b + x^{1-a}y^{1-b}
=
E_{a,b},
\end{equation}
and therefore
\begin{equation}
\pi\!\left(u+xy\,u^{-1}\right)=1.
\label{eq:RHS-parity-blocked}
\end{equation}
Equations \eqref{eq:LHS-parity-blocked} and \eqref{eq:RHS-parity-blocked} contradict Eq.~\eqref{eq:key-eqn-blocked}. Therefore no Hermitian matrix $C$ satisfying Eq.~\eqref{eq:Cd-blocked}, accordingly no matrix $Q$ obeying Eqs.~\eqref{eq:Q-symplectic-blocked}--\eqref{eq:Q-duality-blocked} exists.
\end{proof}

Theorem 1 shows that an exact Clifford implementation of electromagnetic duality in the square-lattice toric code cannot have order two whenever the symmetry is invariant under an odd $l\times l$ enlarged translation group. Therefore any such Clifford electromagnetic duality must have order at least four, consistent with the exact $\mathbb{Z}_4$ symmetry discussed in Refs.~\cite{Barkeshli:2022wuz, shirley2025QCA}. Concretely, one exact $\mathbb{Z}_4$ Clifford symmetry is represented by
\begin{align}
    Q_{\mathbb{Z}_4} =  \begin{pmatrix}
        0 & 0 & 1+x & 1+\bar x + y \\
        0 & 0 & x\bar y & 1+\bar y  \\
        1+y & \bar x y & 0 & 0 \\
        1+x+\bar y & 1+\bar x & 0 & 0
    \end{pmatrix},
\end{align}
which satisfies
\begin{align*}
Q_{\mathbb{Z}_4}^\dagger \Lambda Q_{\mathbb{Z}_4} &= \Lambda, \\
(Q_{\mathbb{Z}_4})^4 &= I, \\
Q_{\mathbb{Z}_4}\sigma_{\mathrm{TC}} &= \sigma_{\mathrm{TC}} S_{U=1}.
\end{align*}

While our theorem is proved for odd $l$, we have also carried out an explicit computer search for the first even supercells, $l=2,4$. In these cases we searched for exact order-two implementations of the electromagnetic duality whose induced action on Pauli operators is local. No solution to Eq.~\eqref{eq:Cd-blocked} was found for $l=2,4$ when the Laurent-polynomial ansatz for each entry of $C$ was truncated to sum of monomials $X^m Y^n$ with $|m|, |n| \le 5$, and $U=x^ay^b$ with $0\le a,b\le 1$. Although this does not yet constitute a proof for all even $l$, it provides strong evidence that exact internal $\mathbb{Z}_2$ electromagnetic duality is absent more generally within Clifford setting.

\textit{Faithful $\mathbb{Z}_4$ electromagnetic symmetry action in $\mathbb{Z}_2$ gauge theory.--}
We close with a brief comment on electromagnetic duality in $\mathbb{Z}_2$ gauge theory in generic dimensions. Our lattice result isolates an obstruction to exact $\Z_2$ symmetry  specific to Clifford implementations, while admitting exact Clifford $\Z_4$ symmetry for electromagnetic duality on the microscopic Hilbert space. The exact $\Z_2$ electromagnetic duality symmetry is realized by a non-Clifford circuit.

On the other hand, in continuum $\mathbb{Z}_2$ gauge theory on a generic closed spatial manifold, electromagnetic duality is generated by a condensation defect of fermionic Wilson surface operators \cite{Roumpedakis2023}, whose action on the Hilbert space is generically faithful $\mathbb{Z}_4$ rather than faithful $\mathbb{Z}_2$ \cite{Chen:2021xuc, Kobayashi2024crosscap}; in particular, its square can differ from the identity by the fermionic topological operator $W_\psi$, depending on the  topology of a spatial manifold (see Supplemental Materials for details). It would be very interesting to understand whether the $\mathbb{Z}_4$ algebra of the Clifford circuit can be derived directly by identifying the Clifford circuits as condensation defects on lattices.

\section*{Acknowledgments}
We thank Yu-An Chen and Po-Shen Hsin for helpful discussions and comments on the draft.
R.K. is supported by the U.S. Department of Energy, Office of Science, Office of High Energy Physics under Award Number DE-SC0009988 and by the Sivian Fund at the Institute for Advanced Study. 

\bibliography{bibliography.bib}

\begin{thebibliography}{23}%
\makeatletter
\providecommand \@ifxundefined [1]{%
 \@ifx{#1\undefined}
}%
\providecommand \@ifnum [1]{%
 \ifnum #1\expandafter \@firstoftwo
 \else \expandafter \@secondoftwo
 \fi
}%
\providecommand \@ifx [1]{%
 \ifx #1\expandafter \@firstoftwo
 \else \expandafter \@secondoftwo
 \fi
}%
\providecommand \natexlab [1]{#1}%
\providecommand \enquote  [1]{``#1''}%
\providecommand \bibnamefont  [1]{#1}%
\providecommand \bibfnamefont [1]{#1}%
\providecommand \citenamefont [1]{#1}%
\providecommand \href@noop [0]{\@secondoftwo}%
\providecommand \href [0]{\begingroup \@sanitize@url \@href}%
\providecommand \@href[1]{\@@startlink{#1}\@@href}%
\providecommand \@@href[1]{\endgroup#1\@@endlink}%
\providecommand \@sanitize@url [0]{\catcode `\\12\catcode `\$12\catcode `\&12\catcode `\#12\catcode `\^12\catcode `\_12\catcode `\%12\relax}%
\providecommand \@@startlink[1]{}%
\providecommand \@@endlink[0]{}%
\providecommand \url  [0]{\begingroup\@sanitize@url \@url }%
\providecommand \@url [1]{\endgroup\@href {#1}{\urlprefix }}%
\providecommand \urlprefix  [0]{URL }%
\providecommand \Eprint [0]{\href }%
\providecommand \doibase [0]{http://dx.doi.org/}%
\providecommand \selectlanguage [0]{\@gobble}%
\providecommand \bibinfo  [0]{\@secondoftwo}%
\providecommand \bibfield  [0]{\@secondoftwo}%
\providecommand \translation [1]{[#1]}%
\providecommand \BibitemOpen [0]{}%
\providecommand \bibitemStop [0]{}%
\providecommand \bibitemNoStop [0]{.\EOS\space}%
\providecommand \EOS [0]{\spacefactor3000\relax}%
\providecommand \BibitemShut  [1]{\csname bibitem#1\endcsname}%
\let\auto@bib@innerbib\@empty
\bibitem [{\citenamefont {Kitaev}(2003)}]{Kitaev2003toriccode}%
  \BibitemOpen
  \bibfield  {author} {\bibinfo {author} {\bibfnamefont {A.}~\bibnamefont {Kitaev}},\ }\href {\doibase 10.1016/s0003-4916(02)00018-0} {\bibfield  {journal} {\bibinfo  {journal} {Annals of Physics}\ }\textbf {\bibinfo {volume} {303}},\ \bibinfo {pages} {2–30} (\bibinfo {year} {2003})}\BibitemShut {NoStop}%
\bibitem [{\citenamefont {Yoshida}(2015)}]{Yoshida_gate_SPT_2015}%
  \BibitemOpen
  \bibfield  {author} {\bibinfo {author} {\bibfnamefont {B.}~\bibnamefont {Yoshida}},\ }\href {\doibase 10.1103/PhysRevB.91.245131} {\bibfield  {journal} {\bibinfo  {journal} {Phys. Rev. B}\ }\textbf {\bibinfo {volume} {91}},\ \bibinfo {pages} {245131} (\bibinfo {year} {2015})}\BibitemShut {NoStop}%
\bibitem [{\citenamefont {Yoshida}(2016)}]{Yoshida:2015cia}%
  \BibitemOpen
  \bibfield  {author} {\bibinfo {author} {\bibfnamefont {B.}~\bibnamefont {Yoshida}},\ }\href {\doibase 10.1103/PhysRevB.93.155131} {\bibfield  {journal} {\bibinfo  {journal} {Phys. Rev. B}\ }\textbf {\bibinfo {volume} {93}},\ \bibinfo {pages} {155131} (\bibinfo {year} {2016})},\ \Eprint {http://arxiv.org/abs/1508.03468} {arXiv:1508.03468 [cond-mat.str-el]} \BibitemShut {NoStop}%
\bibitem [{\citenamefont {Yoshida}(2017)}]{Yoshida2017387}%
  \BibitemOpen
  \bibfield  {author} {\bibinfo {author} {\bibfnamefont {B.}~\bibnamefont {Yoshida}},\ }\href {\doibase https://doi.org/10.1016/j.aop.2016.12.014} {\bibfield  {journal} {\bibinfo  {journal} {Annals of Physics}\ }\textbf {\bibinfo {volume} {377}},\ \bibinfo {pages} {387} (\bibinfo {year} {2017})}\BibitemShut {NoStop}%
\bibitem [{\citenamefont {Barkeshli}\ \emph {et~al.}(2023)\citenamefont {Barkeshli}, \citenamefont {Chen}, \citenamefont {Huang}, \citenamefont {Kobayashi}, \citenamefont {Tantivasadakarn},\ and\ \citenamefont {Zhu}}]{Barkeshli:2022wuz}%
  \BibitemOpen
  \bibfield  {author} {\bibinfo {author} {\bibfnamefont {M.}~\bibnamefont {Barkeshli}}, \bibinfo {author} {\bibfnamefont {Y.-A.}\ \bibnamefont {Chen}}, \bibinfo {author} {\bibfnamefont {S.-J.}\ \bibnamefont {Huang}}, \bibinfo {author} {\bibfnamefont {R.}~\bibnamefont {Kobayashi}}, \bibinfo {author} {\bibfnamefont {N.}~\bibnamefont {Tantivasadakarn}}, \ and\ \bibinfo {author} {\bibfnamefont {G.}~\bibnamefont {Zhu}},\ }\href {\doibase 10.21468/SciPostPhys.14.4.065} {\bibfield  {journal} {\bibinfo  {journal} {SciPost Phys.}\ }\textbf {\bibinfo {volume} {14}},\ \bibinfo {pages} {065} (\bibinfo {year} {2023})},\ \Eprint {http://arxiv.org/abs/2208.07367} {arXiv:2208.07367 [cond-mat.str-el]} \BibitemShut {NoStop}%
\bibitem [{\citenamefont {Barkeshli}\ \emph {et~al.}(2024{\natexlab{a}})\citenamefont {Barkeshli}, \citenamefont {Chen}, \citenamefont {Hsin},\ and\ \citenamefont {Kobayashi}}]{Barkeshli:2022edm}%
  \BibitemOpen
  \bibfield  {author} {\bibinfo {author} {\bibfnamefont {M.}~\bibnamefont {Barkeshli}}, \bibinfo {author} {\bibfnamefont {Y.-A.}\ \bibnamefont {Chen}}, \bibinfo {author} {\bibfnamefont {P.-S.}\ \bibnamefont {Hsin}}, \ and\ \bibinfo {author} {\bibfnamefont {R.}~\bibnamefont {Kobayashi}},\ }\href {\doibase 10.21468/SciPostPhys.16.4.089} {\bibfield  {journal} {\bibinfo  {journal} {SciPost Phys.}\ }\textbf {\bibinfo {volume} {16}},\ \bibinfo {pages} {089} (\bibinfo {year} {2024}{\natexlab{a}})},\ \Eprint {http://arxiv.org/abs/2211.11764} {arXiv:2211.11764 [cond-mat.str-el]} \BibitemShut {NoStop}%
\bibitem [{\citenamefont {Barkeshli}\ \emph {et~al.}(2024{\natexlab{b}})\citenamefont {Barkeshli}, \citenamefont {Hsin},\ and\ \citenamefont {Kobayashi}}]{Barkeshli:2023bta}%
  \BibitemOpen
  \bibfield  {author} {\bibinfo {author} {\bibfnamefont {M.}~\bibnamefont {Barkeshli}}, \bibinfo {author} {\bibfnamefont {P.-S.}\ \bibnamefont {Hsin}}, \ and\ \bibinfo {author} {\bibfnamefont {R.}~\bibnamefont {Kobayashi}},\ }\href {\doibase 10.21468/SciPostPhys.16.5.122} {\bibfield  {journal} {\bibinfo  {journal} {SciPost Phys.}\ }\textbf {\bibinfo {volume} {16}},\ \bibinfo {pages} {122} (\bibinfo {year} {2024}{\natexlab{b}})},\ \Eprint {http://arxiv.org/abs/2311.05674} {arXiv:2311.05674 [cond-mat.str-el]} \BibitemShut {NoStop}%
\bibitem [{\citenamefont {Shirley}\ \emph {et~al.}(2025)\citenamefont {Shirley}, \citenamefont {Zhang}, \citenamefont {Ji},\ and\ \citenamefont {Levin}}]{shirley2025QCA}%
  \BibitemOpen
  \bibfield  {author} {\bibinfo {author} {\bibfnamefont {W.}~\bibnamefont {Shirley}}, \bibinfo {author} {\bibfnamefont {C.}~\bibnamefont {Zhang}}, \bibinfo {author} {\bibfnamefont {W.}~\bibnamefont {Ji}}, \ and\ \bibinfo {author} {\bibfnamefont {M.}~\bibnamefont {Levin}},\ }\href {https://arxiv.org/abs/2507.21267} {\  (\bibinfo {year} {2025})},\ \Eprint {http://arxiv.org/abs/2507.21267} {arXiv:2507.21267 [cond-mat.str-el]} \BibitemShut {NoStop}%
\bibitem [{\citenamefont {Tu}\ \emph {et~al.}(2026)\citenamefont {Tu}, \citenamefont {Long},\ and\ \citenamefont {Else}}]{Tu_2026}%
  \BibitemOpen
  \bibfield  {author} {\bibinfo {author} {\bibfnamefont {Y.-T.}\ \bibnamefont {Tu}}, \bibinfo {author} {\bibfnamefont {D.~M.}\ \bibnamefont {Long}}, \ and\ \bibinfo {author} {\bibfnamefont {D.~V.}\ \bibnamefont {Else}},\ }\href {\doibase 10.1103/m188-w1ct} {\bibfield  {journal} {\bibinfo  {journal} {Physical Review X}\ }\textbf {\bibinfo {volume} {16}} (\bibinfo {year} {2026}),\ 10.1103/m188-w1ct}\BibitemShut {NoStop}%
\bibitem [{\citenamefont {Haah}(2013)}]{Haah2013polynomial}%
  \BibitemOpen
  \bibfield  {author} {\bibinfo {author} {\bibfnamefont {J.}~\bibnamefont {Haah}},\ }\href {\doibase 10.1007/s00220-013-1810-2} {\bibfield  {journal} {\bibinfo  {journal} {Communications in Mathematical Physics}\ }\textbf {\bibinfo {volume} {324}},\ \bibinfo {pages} {351–399} (\bibinfo {year} {2013})}\BibitemShut {NoStop}%
\bibitem [{\citenamefont {Haah}(2021)}]{Haah2021classification}%
  \BibitemOpen
  \bibfield  {author} {\bibinfo {author} {\bibfnamefont {J.}~\bibnamefont {Haah}},\ }\href {\doibase 10.1063/5.0021068} {\bibfield  {journal} {\bibinfo  {journal} {Journal of Mathematical Physics}\ }\textbf {\bibinfo {volume} {62}} (\bibinfo {year} {2021}),\ 10.1063/5.0021068}\BibitemShut {NoStop}%
\bibitem [{\citenamefont {Haah}(2011)}]{Haah2011haahcode}%
  \BibitemOpen
  \bibfield  {author} {\bibinfo {author} {\bibfnamefont {J.}~\bibnamefont {Haah}},\ }\href {\doibase 10.1103/physreva.83.042330} {\bibfield  {journal} {\bibinfo  {journal} {Physical Review A}\ }\textbf {\bibinfo {volume} {83}} (\bibinfo {year} {2011}),\ 10.1103/physreva.83.042330}\BibitemShut {NoStop}%
\bibitem [{\citenamefont {Haah}\ \emph {et~al.}(2022)\citenamefont {Haah}, \citenamefont {Fidkowski},\ and\ \citenamefont {Hastings}}]{Haah2022QCA}%
  \BibitemOpen
  \bibfield  {author} {\bibinfo {author} {\bibfnamefont {J.}~\bibnamefont {Haah}}, \bibinfo {author} {\bibfnamefont {L.}~\bibnamefont {Fidkowski}}, \ and\ \bibinfo {author} {\bibfnamefont {M.~B.}\ \bibnamefont {Hastings}},\ }\href {\doibase 10.1007/s00220-022-04528-1} {\bibfield  {journal} {\bibinfo  {journal} {Communications in Mathematical Physics}\ }\textbf {\bibinfo {volume} {398}},\ \bibinfo {pages} {469–540} (\bibinfo {year} {2022})}\BibitemShut {NoStop}%
\bibitem [{\citenamefont {Tantivasadakarn}(2020)}]{Tantivasadakarn2020JW}%
  \BibitemOpen
  \bibfield  {author} {\bibinfo {author} {\bibfnamefont {N.}~\bibnamefont {Tantivasadakarn}},\ }\href {\doibase 10.1103/physrevresearch.2.023353} {\bibfield  {journal} {\bibinfo  {journal} {Physical Review Research}\ }\textbf {\bibinfo {volume} {2}} (\bibinfo {year} {2020}),\ 10.1103/physrevresearch.2.023353}\BibitemShut {NoStop}%
\bibitem [{\citenamefont {Liang}\ \emph {et~al.}(2025)\citenamefont {Liang}, \citenamefont {Yang}, \citenamefont {Iosue},\ and\ \citenamefont {Chen}}]{liang2025operatoralgebraalgorithmicconstruction}%
  \BibitemOpen
  \bibfield  {author} {\bibinfo {author} {\bibfnamefont {Z.}~\bibnamefont {Liang}}, \bibinfo {author} {\bibfnamefont {B.}~\bibnamefont {Yang}}, \bibinfo {author} {\bibfnamefont {J.~T.}\ \bibnamefont {Iosue}}, \ and\ \bibinfo {author} {\bibfnamefont {Y.-A.}\ \bibnamefont {Chen}},\ }\href {https://arxiv.org/abs/2410.11942} {\enquote {\bibinfo {title} {Operator algebra and algorithmic construction of boundaries and defects in (2+1)d topological pauli stabilizer codes},}\ } (\bibinfo {year} {2025}),\ \Eprint {http://arxiv.org/abs/2410.11942} {arXiv:2410.11942 [quant-ph]} \BibitemShut {NoStop}%
\bibitem [{\citenamefont {Chen}\ \emph {et~al.}(2024)\citenamefont {Chen}, \citenamefont {Gorshkov},\ and\ \citenamefont {Xu}}]{Chen2024fermionic}%
  \BibitemOpen
  \bibfield  {author} {\bibinfo {author} {\bibfnamefont {Y.-A.}\ \bibnamefont {Chen}}, \bibinfo {author} {\bibfnamefont {A.~V.}\ \bibnamefont {Gorshkov}}, \ and\ \bibinfo {author} {\bibfnamefont {Y.}~\bibnamefont {Xu}},\ }\href {\doibase 10.21468/scipostphys.16.1.033} {\bibfield  {journal} {\bibinfo  {journal} {SciPost Physics}\ }\textbf {\bibinfo {volume} {16}} (\bibinfo {year} {2024}),\ 10.21468/scipostphys.16.1.033}\BibitemShut {NoStop}%
\bibitem [{\citenamefont {Sun}\ \emph {et~al.}(2026)\citenamefont {Sun}, \citenamefont {Yang}, \citenamefont {Wang}, \citenamefont {Tantivasadakarn},\ and\ \citenamefont {Chen}}]{sun2026cliffordquantumcellularautomata}%
  \BibitemOpen
  \bibfield  {author} {\bibinfo {author} {\bibfnamefont {M.}~\bibnamefont {Sun}}, \bibinfo {author} {\bibfnamefont {B.}~\bibnamefont {Yang}}, \bibinfo {author} {\bibfnamefont {Z.}~\bibnamefont {Wang}}, \bibinfo {author} {\bibfnamefont {N.}~\bibnamefont {Tantivasadakarn}}, \ and\ \bibinfo {author} {\bibfnamefont {Y.-A.}\ \bibnamefont {Chen}},\ }\href {https://arxiv.org/abs/2509.07099} {\enquote {\bibinfo {title} {Clifford quantum cellular automata from topological quantum field theories and invertible subalgebras},}\ } (\bibinfo {year} {2026}),\ \Eprint {http://arxiv.org/abs/2509.07099} {arXiv:2509.07099 [math.QA]} \BibitemShut {NoStop}%
\bibitem [{\citenamefont {Ruba}\ and\ \citenamefont {Yang}(2025)}]{ruba2025wittgroupsbulkboundarycorrespondence}%
  \BibitemOpen
  \bibfield  {author} {\bibinfo {author} {\bibfnamefont {B.}~\bibnamefont {Ruba}}\ and\ \bibinfo {author} {\bibfnamefont {B.}~\bibnamefont {Yang}},\ }\href {https://arxiv.org/abs/2509.10418} {\enquote {\bibinfo {title} {Witt groups and bulk-boundary correspondence for stabilizer states},}\ } (\bibinfo {year} {2025}),\ \Eprint {http://arxiv.org/abs/2509.10418} {arXiv:2509.10418 [math-ph]} \BibitemShut {NoStop}%
\bibitem [{\citenamefont {Roumpedakis}\ \emph {et~al.}(2023)\citenamefont {Roumpedakis}, \citenamefont {Seifnashri},\ and\ \citenamefont {Shao}}]{Roumpedakis2023}%
  \BibitemOpen
  \bibfield  {author} {\bibinfo {author} {\bibfnamefont {K.}~\bibnamefont {Roumpedakis}}, \bibinfo {author} {\bibfnamefont {S.}~\bibnamefont {Seifnashri}}, \ and\ \bibinfo {author} {\bibfnamefont {S.-H.}\ \bibnamefont {Shao}},\ }\href {\doibase 10.1007/s00220-023-04706-9} {\bibfield  {journal} {\bibinfo  {journal} {Communications in Mathematical Physics}\ }\textbf {\bibinfo {volume} {401}},\ \bibinfo {pages} {3043–3107} (\bibinfo {year} {2023})}\BibitemShut {NoStop}%
\bibitem [{\citenamefont {Chen}\ \emph {et~al.}(2021)\citenamefont {Chen}, \citenamefont {Dua}, \citenamefont {Hsin}, \citenamefont {Jian}, \citenamefont {Shirley},\ and\ \citenamefont {Xu}}]{Chen:2021xuc}%
  \BibitemOpen
  \bibfield  {author} {\bibinfo {author} {\bibfnamefont {X.}~\bibnamefont {Chen}}, \bibinfo {author} {\bibfnamefont {A.}~\bibnamefont {Dua}}, \bibinfo {author} {\bibfnamefont {P.-S.}\ \bibnamefont {Hsin}}, \bibinfo {author} {\bibfnamefont {C.-M.}\ \bibnamefont {Jian}}, \bibinfo {author} {\bibfnamefont {W.}~\bibnamefont {Shirley}}, \ and\ \bibinfo {author} {\bibfnamefont {C.}~\bibnamefont {Xu}},\ }\href@noop {} {\  (\bibinfo {year} {2021})},\ \Eprint {http://arxiv.org/abs/2112.02137} {arXiv:2112.02137 [cond-mat.str-el]} \BibitemShut {NoStop}%
\bibitem [{\citenamefont {Kobayashi}\ and\ \citenamefont {Zhu}(2024)}]{Kobayashi2024crosscap}%
  \BibitemOpen
  \bibfield  {author} {\bibinfo {author} {\bibfnamefont {R.}~\bibnamefont {Kobayashi}}\ and\ \bibinfo {author} {\bibfnamefont {G.}~\bibnamefont {Zhu}},\ }\href {\doibase 10.1103/prxquantum.5.020360} {\bibfield  {journal} {\bibinfo  {journal} {PRX Quantum}\ }\textbf {\bibinfo {volume} {5}} (\bibinfo {year} {2024}),\ 10.1103/prxquantum.5.020360}\BibitemShut {NoStop}%
\bibitem [{\citenamefont {Kaidi}\ \emph {et~al.}(2023)\citenamefont {Kaidi}, \citenamefont {Ohmori},\ and\ \citenamefont {Zheng}}]{kaidi2023symmetrytftsnoninvertibledefects}%
  \BibitemOpen
  \bibfield  {author} {\bibinfo {author} {\bibfnamefont {J.}~\bibnamefont {Kaidi}}, \bibinfo {author} {\bibfnamefont {K.}~\bibnamefont {Ohmori}}, \ and\ \bibinfo {author} {\bibfnamefont {Y.}~\bibnamefont {Zheng}},\ }\href {https://arxiv.org/abs/2209.11062} {\enquote {\bibinfo {title} {Symmetry tfts for non-invertible defects},}\ } (\bibinfo {year} {2023}),\ \Eprint {http://arxiv.org/abs/2209.11062} {arXiv:2209.11062 [hep-th]} \BibitemShut {NoStop}%
\bibitem [{\citenamefont {{Karlheinz Knapp}}()}]{ManifoldAtlasWu}%
  \BibitemOpen
  \bibfield  {author} {\bibinfo {author} {\bibnamefont {{Karlheinz Knapp}}},\ }\href {http://www.map.mpim-bonn.mpg.de/Wu_class} {\enquote {\bibinfo {title} {Wu class},}\ }\BibitemShut {NoStop}%
\end{thebibliography}%

\onecolumngrid

\vspace{0.3cm}

\newpage

\begin{center}
\Large{\bf Supplemental Materials}
\end{center}
\onecolumngrid

\section{Condensation defect of $\Z_2$ gauge theory and electromagnetic duality}

Here we recall that the electromagnetic duality of $\Z_2$ gauge theory in generic dimensions is generated by a condensation defect, and that the condensation defect generates faithful $\Z_4$ action on the Hilbert space.

Consider a $\Z_2$ gauge theory
\begin{equation}
S=\pi\int adb ,
\end{equation}
where $a$ and $b$ are $\mathbb{Z}_2$-valued $k$-form gauge fields.
For every closed $k$-cycle $\Sigma_k$, the theory has electric and magnetic topological operators
\begin{equation}
W_e(\Sigma_k)=(-1)^{\int_{\Sigma_k} a},
\qquad
W_m(\Sigma_k)=(-1)^{\int_{\Sigma_k} b},
\end{equation}
as well as the fermionic operator
\begin{equation}
W_\psi(\Sigma_k)=W_e(\Sigma_k)W_m(\Sigma_k),
\end{equation}
which generate $\Z_2$ $k$-form symmetry.
Electromagnetic duality exchanges $W_e$ and $W_m$.

A topological operator implementing the electromagnetic duality is given by the condensation defect of the fermionic operator (see Refs.~\cite{Roumpedakis2023}, \cite{kaidi2023symmetrytftsnoninvertibledefects} for discussions in (2+1)D and (4+1)D),
\begin{equation}
U_{\rm em}
=
\frac{1}{\sqrt{|H_k(M,\mathbb{Z}_2)|}}
\sum_{\Sigma\in H_k(M,\mathbb{Z}_2)}
W_\psi(\Sigma),
\end{equation}
with $M$ a spatial $2k$-manifold. This operator indeed permutes $e$ and $m$:
\begin{equation}
U_{\rm em} W_e(\Sigma_k)=W_m(\Sigma_k)U_{\rm em},\ 
U_{\rm em} W_m(\Sigma_k)=W_e(\Sigma_k)U_{\rm em}.
\end{equation}
The square of this duality operator is not generically trivial, but rather
\begin{equation}
U_{\rm em}^2 = W_\psi\!\bigl(\mathrm{PD}(\nu_k)\bigr),
\label{eq:Uem-square}
\end{equation}
where $\nu_k$ is the $k$th Wu class of the spatial manifold \cite{ManifoldAtlasWu} and $\mathrm{PD}$ denotes Poincar\'e duality. This is derived by
\begin{align}
    \begin{split}
        U_{\rm em}^2 &= \frac{1}{|H_k(M,\Z_2)|}\sum_{\Sigma,\Sigma'} W_{\psi}(\Sigma)W_{\psi}(\Sigma') \\
        &= \frac{1}{|H_k(M,\Z_2)|}\sum_{\Sigma,\Sigma'} W_{\psi}(\Sigma+\Sigma') (-1)^{\int\sigma\cup\sigma'} \\
        &= \frac{1}{|H_k(M,\Z_2)|}\sum_{\Sigma,\Sigma'} W_{\psi}(\Sigma) (-1)^{\int\sigma'\cup(\sigma+\sigma')} \\
        &= \frac{1}{|H_k(M,\Z_2)|}\sum_{\Sigma,\Sigma'} W_{\psi}(\Sigma) (-1)^{\int\sigma'\cup(\sigma+\nu_k)} \\
        &= W_{\psi}(\text{PD}(\nu_k))~, \\
    \end{split}
\end{align}
where we used $\int_{M} \sigma'\cup\sigma' = \int_M \nu_k\cup \sigma'$ mod 2.
It follows immediately that
\begin{equation}
U_{\rm em}^4=1.
\end{equation}
Thus electromagnetic duality acts faithfully by the $\mathbb{Z}_4$ action on a generic spatial manifold. In other words, the $\mathbb{Z}_2$ electromagnetic symmetry of $\mathbb{Z}_2$ gauge theory is extended by the $k$-form $\mathbb{Z}_2$ symmetry generated by $W_\psi$. For instance, the operator $U_{\rm em}$ generates the $\Z_4$ action in (2+1)D $\Z_2$ gauge theory when the spatial manifold is a 2d real projective plane $\mathbb{RP}^2$ \cite{Kobayashi2024crosscap}.
In generic dimensions, the above algebra $U_{\rm em}^2 = W_{\psi}$ signals a higher-group symmetry structure, as pointed out in Ref.~\cite{Chen:2021xuc} in (4+1)D.

\vfill

\end{document}